\documentclass[pra,twocolumn,showpacs,10pt,floatfix]{revtex4-1}

\usepackage{amsmath}
\usepackage{amsfonts}
\usepackage{amssymb}
\usepackage{mathrsfs}
\usepackage{latexsym}
\usepackage{graphicx}

\newtheorem{theorem}{Theorem}

\newenvironment{proof}[1][Proof]{\noindent\textbf{#1.} }{\ \rule{0.5em}{0.5em}}

\begin{document}

\title{Monogamy of Quantum Discord by Multipartite Correlations} 

\author{H. C. \surname{Braga}}
\email{helenacbraga@gmail.com}
\author{C. C. \surname{Rulli}}
\author{Thiago R. de \surname{Oliveira}}
\author{M. S. \surname{Sarandy}}

\affiliation{Instituto de F\'{\i}sica, Universidade Federal Fluminense, Av. Gal. Milton Tavares de Souza s/n, Gragoat\'a, 
24210-346, Niter\'oi, RJ, Brazil.}

\date{\today }

\begin{abstract}
We introduce a monogamy inequality for quantum correlations, which implies that the sum of pairwise quantum correlations 
is upper limited by the amount of multipartite quantum correlations as measured by the global quantum discord. This monogamy bound 
holds either for pure or mixed quantum states provided that bipartite quantum discord does not increase under discard of subsystems. 
We illustrate the monogamy behavior for multipartite pure states with Schmidt decomposition as well as for W-GHZ mixed states. 
As a by-product, we apply the monogamy bound to investigate residual multipartite 
correlations.

\end{abstract}


\maketitle

\section{Introduction}

Quantum discord (QD)~\cite{Ollivier:01} has recently been identified as a general resource in quantum 
information protocols~(see, e.g., Refs.~\cite{Madhok:12,Gu:12}). In quantum computation, it has been 
conjectured as the origin of speed up in the deterministic quantum computation with one qubit (DQC1) 
mixed-state model~\cite{DQC1}. Moreover, a fundamental role has been attributed to QD in tasks such as 
quantum locking~\cite{locking} and quantum state discrimination~\cite{Roa:11}. Besides quantum protocols, 
remarkable applications of QD have also been found in the characterization of quantum phase transitions~\cite{QPT} 
and in the description of quantum dynamics under decoherence~\cite{Decoh}. 

In order to use quantum correlations (as provided by QD) as a resource, we are faced with the problem of their 
distribution throughout a multipartite state. In this context, a monogamous behavior may reveal important information 
about the structure of quantum correlations. For instance, 
monogamy has been found to be the essential feature allowing for security in quantum key distribution~\cite{Pawlowski:04}. In general grounds, 
it has been investigated whether any given measure ${\cal Q}$ of bipartite quantum correlation can obey a monogamy bound in 
an arbitrary tripartite quantum system $ABC$. A measure ${\cal Q}$ has typically been defined as monogamous if it follows the 
inequality 
\begin{equation}
{\cal Q} (A,BC) \ge {\cal Q}(A,B) + {\cal Q}(A,C).
\label{mon1}
\end{equation}
Therefore, given a fixed value for the quantum correlation ${\cal Q} (A,BC)$ between subsystem $A$ and the remaining part $BC$ of the 
system, then $A$ cannot be freely correlated with the individual subsystems $B$ and $C$. Such a monogamy bound is obeyed by certain 
entanglement measures, which is indeed the origin of {\it tangle}~\cite{Coffman:00} as a measure of 
genuine multipartite entanglement. Whether or not monogamy can be obtained for a quantum correlation measure in arbitrary (pure or 
mixed) quantum states poses therefore as a further challenge. However, many attempts turned out to find a polygamous behavior for  
QD~\cite{Prabhu:11,Giorgi-1:11,Ren:11,Fanchini:11,Bai:12}. More generally, such a negative answer has been obtained for any sensible 
measure ${\cal Q}$ of quantum correlations~\cite{Streltsov:11} (even though a function of ${\cal Q}$ may be able to yield a 
state-dependent trade-off~\cite{Salini:12}). 

The violation of the monogamy inequality~(\ref{mon1}) means that the {\it single-site correlation} ${\cal Q} (A,BC)$ is unable in general to 
set a limit for the sum of pairwise correlations. However, it does {\it not} imply that subsystems can be freely correlated. 
More specifically, it is possible to restore monogamy if we can find out a physical quantity that is able to provide a bound 
for the sum of pairwise correlations. {\it Here, we will investigate this problem by exploring such a route, 
providing a monogamy bound that holds for either pure or mixed states in a multipartite system composed of a number $N$ of subsystems}. 
In order to achieve this aim, we will show that a limit for the sum of pairwise quantum correlations is provided by 
a multipartite extension of QD, named as {\it global quantum discord} (GQD), which has been proposed in Ref.~\cite{Rulli:11}. 
In the case of tripartite states, a monogamy bound will then be obtained by replacing ${\cal Q} (A,BC)$ for the GQD, 
denoted by $\mathcal{D}\left( A : B : C \right)$, yielding
\begin{equation}
\mathcal{D}\left( A : B : C \right) \ge \mathcal{D}\left( A : B \right) + \mathcal{D}\left( A : C \right).
\label{mon2}
\end{equation}
\begin{figure}[th]
\centering {\includegraphics[angle=0,scale=0.3]{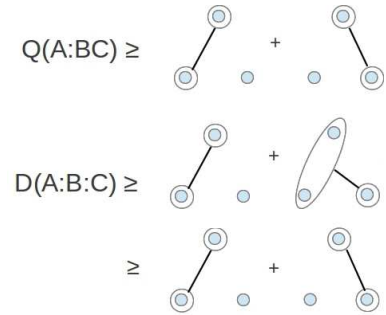}}
\caption{(Color online) Schematic description of tripartite monogamy structure for ${\cal Q} (A,BC)$ and 
$\mathcal{D}\left( A : B : C \right)$.}
\label{f1}
\end{figure}

A schematic description comparing the monogamy bounds for ${\cal Q} (A,BC)$ and $\mathcal{D}\left( A : B : C \right)$ 
is displayed in Fig.~\ref{f1}. The validity of inequality (\ref{mon2}) will be shown for all quantum states whose bipartite 
quantum discords do not increase under discard of subsystems. Moreover, we will show that GQD allows for the 
extension of Eq.~(\ref{mon2}) for the case of $N$ subsystems. We observe that the maximum of GQD typically increases 
with $N$, which accounts for the fact that pairwise correlations are less monogamous than entanglement. On the other hand, 
as will be illustrated in several examples, GQD is sufficiently limited to provide a relevant constraint for the sum of 
pairwise QD.  As an application, we will use Eq.~(\ref{mon2}) to propose a natural measure of residual multipartite quantum 
correlations.

\section{Quantum discord and its multipartite extension}

QD, which has been introduced in Ref.~\cite{Ollivier:01} and provides a measure of the quantumness of correlations, 
can be suitably defined as the minimal loss of total correlation after a non-selective measurement~\cite{Luo:10}. 
Indeed, consider a bipartite system $AB$ composed of subsystems $A$ and $B$. Denoting by $\hat{\rho}_{AB}$ the density 
operator of $AB$ and by $\hat{\rho}_A$ and $\hat{\rho}_B$ the density operators of parts $A$ and $B$, respectively, 
the total correlation between $A$ and $B$ is measured by the quantum mutual information~\cite{Groisman:05} 
\begin{equation}
I(\hat{\rho}_{AB}) = S(\hat{\rho}_A) - S(\hat{\rho}_A | \hat{\rho}_B),
\end{equation}
where $S(\hat{\rho}_A) = -{\textrm{Tr}} \hat{\rho}_A \log_2 \hat{\rho}_A$ is the von Neumann entropy for $A$ and 
\begin{equation}
S(\hat{\rho}_A | \hat{\rho}_B) = S(\hat{\rho}_{AB}) - S(\hat{\rho}_B)
\end{equation}
is the entropy of $A$ conditional on $B$. By operating on subsystem $B$, QD can then be expressed as
$\mathcal{D}\left( A | B \right) = 
\min_{\{\hat{\Pi}_{B}^{k}\}} \left[
I(\hat{\rho}_{AB}) - 
I(\Phi _{B}\left( \hat{\rho}_{AB}\right))\right]$,
where $\Phi _{B}\left( \hat{\rho}_{AB}\right)$ denotes a non-selective measurement $\{\hat{\Pi}_{B}^{j}\}$ on part $B$ of 
$\hat{\rho}_{AB}$, which yields
$\Phi _{B}\left( \hat{\rho}_{AB}\right) = \sum_{j}
\left( \hat{1}_{A}\otimes \hat{\Pi}_{B}^{j} \right) 
\hat{\rho}_{AB}
\left(\hat{1}_{A}\otimes \hat{\Pi}_{B}^{j}\right)$.
Note that this definition is asymmetric with respect to measurement on $A$ or $B$. In particular, 
a strictly classical state requires both $\mathcal{D}\left( A | B \right) = 0$ and $\mathcal{D}\left( B | A \right) = 0$. 
Indeed, this corresponds to a density operator $\hat{\rho}_{AB} = \sum_{i,j} p_{ij} |i\rangle\langle i| \otimes |j\rangle\langle j|$, 
where $p_{ij}$ is a joint probability distribution and the sets $\{|i\rangle\}$ and $\{|j\rangle\}$ constitute orthonormal bases for 
the systems $A$ and $B$, respectively. Such strictly classical states can also be identified by a single measure, which is the symmetrized 
version of QD
\begin{equation}
\mathcal{D}\left( A : B \right) = 
\min_{\{\hat{\Pi}_{A}^{j}\otimes\hat{\Pi}_{B}^{k}\}} \left[
I(\hat{\rho}_{AB}) - 
I(\Phi _{AB}\left( \hat{\rho}_{AB}\right))\right], 
\label{BipartiteDiscord}
\end{equation}
where the measurement operator $\Phi _{AB}$ is given by
\begin{equation}
\Phi _{AB}\left( \hat{\rho}_{AB}\right) =\sum_{j,k} \left(\hat{\Pi}_{A}^{j}\otimes 
\hat{\Pi}_{B}^{k} \right) \hat{\rho}_{AB} \left(\hat{\Pi}_{A}^{j}\otimes \hat{\Pi}_{B}^{k}\right) \, .
\end{equation}
Eq.~(\ref{BipartiteDiscord}) provides the symmetric QD considered in Ref.~\cite{Maziero:10} and 
experimentally witnessed in Refs.~\cite{Auccaise:11,Aguilar:12}. The vanishing of $\mathcal{D}\left( A : B \right)$ 
occurs if and only if the state is classical. In particular, the absence of $\mathcal{D}\left( A : B \right)$ 
is the key ingredient for local sharing of pre-established correlations (local broadcasting)~\cite{Piani:08}. 

Generalizations of quantum discord to multipartite states have been considered in different scenarios~\cite{Modi:11,Chakrabarty:10,Okrasa:11,Rulli:11}, 
which intend to account for quantum correlations that may exist beyond pairwise subsystems in a composite system. 
In this direction, one possible approach to account multipartite quantum correlations is to start from the symmetrized QD and then 
to systematically extend it to the multipartite scenario. This originates GQD as a measure of global quantum discord, 
as proposed in Ref.~\cite{Rulli:11}. The GQD $\mathcal{D}\left( A_1 : \cdots :A_N \right)$ for 
an arbitrary multipartite state $\hat{\rho}_{A_1 \cdots A_N}$ under a set of local 
measurements $\{\hat{\Pi}_{A_1}^{j_1} \otimes \cdots \otimes \hat{\Pi}_{A_N}^{j_N}\}$ can suitably be expressed as~\cite{Rulli:11,Celeri:11,Xu:12}
\begin{equation}
\mathcal{D}\left( A_1 : \cdots : A_N \right) = \min_{\Phi} \left[\,
I(\hat{\rho}_{A_1 \cdots A_N}) - I(\Phi\left( \hat{\rho}_{A_1 \cdots A_N} \right)) \right], 
\label{gqd-def}
\end{equation}
where 
\begin{equation}
\Phi\left( \hat{\rho}_{A_1 \cdots A_N} \right) = \sum_{k} {\hat{\Pi}}_{k} \, \hat{\rho}_{A_1 \cdots A_N} \, 
{\hat{\Pi}}_{k},
\end{equation}
with $\hat{\Pi}_{k} = \hat{\Pi}_{A_1}^{j_1} \otimes \cdots \otimes \hat{\Pi}_{A_N}^{j_N}$ and 
$k$ denoting the index string $(j_1 \cdots j_N$). In Eq.~(\ref{gqd-def}), the mutual information $I(\hat{\rho}_{A_1 \cdots A_N})$ 
is defined following the natural generalization proposed in Ref.~\cite{Groisman:05}, which implies that $I(\hat{\rho}_{A_1 \cdots A_N})$ 
and $I(\Phi\left( \hat{\rho}_{A_1 \cdots A_N} \right))$ are given by
\begin{eqnarray}
\hspace{-0.4cm}I(\hat{\rho}_{A_1 \cdots A_N}) &=& \sum_{k=1}^{N} S\left( \hat{\rho}_{A_k}\right) - 
S\left( \hat{\rho}_{A_1 \cdots A_N} \right) , \nonumber \\
\hspace{-0.4cm}I(\Phi\left( \hat{\rho}_{A_1 \cdots A_N} \right)) &=& \sum_{k=1}^{N} S\left( \Phi \left( \hat{\rho}_{A_k}\right) \right) 
- S\left(\Phi \left( \hat{\rho}_{A_1 \cdots A_N}\right) \right) , 
\label{Imulti}
\end{eqnarray}
where 
\begin{equation}
\Phi \left( \hat{\rho}_{A_k} \right) = \sum_{k^\prime} \hat{\Pi}_{A_k}^{k^\prime} \, \hat{\rho}_{A_k} \,
\hat{\Pi}_{A_k}^{k^\prime}.
\end{equation}
GQD is symmetric with respect to subsystem exchange and shown to be non-negative for arbitrary states~\cite{Rulli:11}. 
Moreover, it can be detected through a convenient (with no extremization procedure) witness operator~\cite{Saguia:11}. 
Concerning its applications, GQD has been shown to be useful in the characterization of quantum phase transitions~\cite{Rulli:11,Campbell:11}. 
In terms of operational interpretation, GQD may play a role in quantum communication, in the sense that its absence means that 
the quantum state simply describes a classical probability multidistribution 
$ \sum_{i_1,\cdots,i_N} p_{i_1 \cdots i_N} |i_1\rangle\langle i_1| \otimes \cdots \otimes |i_N\rangle\langle i_N|$ 
(with $p_{i_1 \cdots i_N} \ge 0$, $\sum p_{i_1 \cdots i_N} = 1$) 
and, therefore, allows for local broadcasting of correlations~\cite{Piani:08}.

\section{Monogamy of quantum correlations and global quantum discord}

Let us begin by defining the loss of correlation $\mathcal{D}_{\Phi} (A_1 : \cdots : A_N)$ in a 
quantum state $\hat{\rho}_{A_1 \cdots A_N}$ generated by a measurement 
$\Phi\left( \hat{\rho}_{A_1 \cdots A_N} \right)$, which is given by 
\begin{equation}
\mathcal{D}_{\Phi}(A_1:\cdots:A_N)  =I\left( \hat{\rho}_{A_1 \cdots A_N} \right) -I\left(
\Phi \left( \hat{\rho}_{A_1 \cdots A_N}\right) \right).
\label{Dphi}
\end{equation}
We can then show that $\mathcal{D}_{\Phi} (A_1 : \cdots : A_N)$ can be decomposed in 
terms of loss of correlation for suitable bipartitions, as provided by the following theorem. 

\begin{theorem}
{\it Given a non-selective measurement $\Phi\left( \hat{\rho}_{A_1 \cdots A_N} \right)$, the 
loss of correlation $\mathcal{D}_{\Phi} (A_1 : \cdots : A_N)$ can be decomposed as} 
\begin{equation}
\mathcal{D}_{\Phi} (A_1 : \cdots : A_N) = \sum_{k=1}^{N-1} \mathcal{D}_{\Phi} (A_1 \cdots A_k : A_{k+1}) .
\label{DN}
\end{equation}
\label{t1}
\end{theorem}
\vspace{-0.4cm}
\begin{proof}
We rewrite $\mathcal{D}_{\Phi} (A_1 : \cdots : A_N)$ by using Eqs.~(\ref{Imulti}) and (\ref{Dphi}), which yields
\begin{eqnarray}
&&\mathcal{D}_{\Phi} (A_1 : \cdots : A_N) = \sum_{k=1}^{N}  \left[ S\left( \hat{\rho}_{A_k}\right) 
- S\left( \Phi_{k}\left( \hat{\rho}_{A_k}\right) \right) \right] \nonumber \\
&&-S\left( \hat{\rho}_{A_1 \cdots A_N} \right) + S\left(\Phi \left( \hat{\rho}_{A_1 \cdots A_N}\right) \right) . 
\label{Dentropy}
\end{eqnarray}
We then add and subtract $S\left( \hat{\rho}_{A_1 \cdots A_{N-1}}\right)$ and 
$S\left(\Phi \left( \hat{\rho}_{A_1 \cdots A_{N-1}}\right)\right)$ in Eq.~(\ref{Dentropy}). After rearrangement of the 
terms, we obtain
\begin{eqnarray}
\mathcal{D}_{\Phi} (A_1 : \cdots : A_N) &=& \mathcal{D}_{\Phi} (A_1 : \cdots : A_{N-1}) + \nonumber \\ 
&& \mathcal{D}_{\Phi} (A_1 \cdots A_{N-1} : A_N) . 
\label{D2}
\end{eqnarray}
By recursively applying Eq.~(\ref{D2}), we can reduce the term $\mathcal{D}_{\Phi} (A_1 : \cdots : A_{N-1})$ 
to a sum of bipartite contributions $\sum_{k=1}^{N-2} \mathcal{D}_{\Phi} (A_1 \cdots A_k : A_{k+1})$ which, 
when added to $\mathcal{D}_{\Phi} (A_1 \cdots A_{N-1} : A_N)$ in Eq.~(\ref{D2}), leads to Eq.~(\ref{DN}). 
\end{proof}

\vspace{0.2cm}

By taking a tripartite system $ABC$, Theorem~\ref{t1} implies that
\begin{equation}
\mathcal{D}_{\Phi} (A : B : C) = \mathcal{D}_{\Phi} (A : B) + \mathcal{D}_{\Phi} (A B : C).
\label{eqdphi}
\end{equation}
Note that Eq.~(\ref{eqdphi}) provides a decomposition of $\mathcal{D}_{\Phi} (A : B : C)$ in terms of bipartite 
contributions. 
As a further step to achieve a monogamy trade-off, we will now introduce GQD by minimizing $\mathcal{D}_{\Phi} (A_1 : \cdots : A_N)$ 
over all measurements $\Phi\left( \hat{\rho}_{A_1 \cdots A_N} \right)$, i.e.,
\begin{equation}
\mathcal{D}\left( A_1 : \cdots : A_N \right) = \min_{\Phi} \left[ \mathcal{D}_{\Phi} (A_1 : \cdots : A_N) \right] .
\end{equation}
Then, a monogamy bound for quantum correlations in $N$-partite systems can be obtained, which is proved below.
\begin{theorem}
{\it For an arbitrary density operator $\hat{\rho}_{A_1 \cdots A_N}$ with an amount of GQD given by  
$\mathcal{D}\left( A_1 : \cdots :A_N \right)$, pairwise QD obeys the monogamy bound }
\begin{equation}
\mathcal{D}\left( A_1 : \cdots :A_N \right) \ge \sum_{k=1}^{N-1} \mathcal{D}\left( A_1 : A_{k+1} \right), 
\label{monogamybound}
\end{equation}
{\it provided that the bipartite QDs $\mathcal{D} (A_1 \cdots A_{k} : A_{k+1})$, with $2\le k <N$), do not increase under 
discard of subsystems, i.e., $\mathcal{D} (A_1 \cdots A_{k} : A_{k+1}) \ge \mathcal{D} (A_1 : A_{k+1})$.}
\label{t3}
\end{theorem}
\begin{proof}
Starting from Theorem~1, we minimize both sides of Eq.~(\ref{DN}) with respect to $\Phi\left( \hat{\rho}_{A_1 \cdots A_N} \right)$, yielding
\begin{equation}
\min_{\Phi} \mathcal{D}_{\Phi} (A_1 : \cdots : A_N) = \min_{\Phi} \sum_{k=1}^{N-1} \mathcal{D}_{\Phi} (A_1 \cdots A_k : A_{k+1}) .
\label{Discord-in1}
\end{equation}
However, we have that
\begin{equation}
\min_{\Phi} \sum_{k=1}^{N-1} \mathcal{D}_{\Phi} (A_1 : A_{k+1}) \ge \sum_{k=1}^{N-1} \min_{\Phi} \mathcal{D}_{\Phi} (A_1 \cdots A_k : A_{k+1}) .
\label{minim}
\end{equation}
Then, by inserting Eq.~(\ref{minim}) into Eq.~(\ref{Discord-in1}), we obtain 
\begin{equation}
\mathcal{D}\left( A_1 : \cdots : A_N \right) \ge  \sum_{k=1}^{N-1} \mathcal{D}\left( A_1 \cdots A_{k} : A_{k+1} \right) .
\end{equation}
Hence, by using the condition that $\mathcal{D} (A_1 \cdots A_{k} : A_{k+1}) \ge \mathcal{D} (A_1 : A_{k+1})$ 
we obtain the monogamy bound stated by inequality~(\ref{monogamybound}).
\end{proof}

\vspace{0.2cm}

In particular, for tripartite states, the monogamy inequality~(\ref{monogamybound}) reduces to inequality~(\ref{mon2}), 
provided that the condition ${\cal D}(A B : C) \ge {\cal D} (A : C)$ is satisfied. 
We observe that this non-increasing behavior of bipartite QD under discard of subsystems is a sufficient condition for 
the validity of monogamy, but it is {\it not} a necessary requirement. This means that monogamy may hold even in more 
general scenarios. For example, consider a tripartite system composed of qubits ABC in the mixed state 
$\rho_{ABC} = (1/2) ( |000 \rangle \langle 000| + |1+1 \rangle \langle 1+1 |)$, 
where $\{|0\rangle,|1\rangle\}$ denotes the computational basis and $|+\rangle = (|0\rangle + |1\rangle)/\sqrt{2}$. 
For this state, we have that QD vanishes for the bipartition A with BC, i.e. $\mathcal{D}(A:BC)=0$. However, by discarding 
subsystem C, we obtain $\rho_{AB} = (1/2) ( |00 \rangle \langle 00| + |1+ \rangle \langle 1+ |)$, which mixes nonorthogonal 
states for subsystem B. We then obtain a nonvanishing QD between A and B, which is given by $\mathcal{D}(A:B) \approx 0.204$. 
However, the monogamy trade-off 
$\mathcal{D}\left( A : B : C \right) \ge \mathcal{D}\left( A : B \right) + \mathcal{D}\left( A : C \right)$ keeps still valid, 
since GQD turns out to be equal to $\mathcal{D}(A:B)$, namely, $\mathcal{D}\left( A : B : C \right) \approx 0.204$ and 
$\mathcal{D}\left( A : C \right) = 0$. Hence, we obtain a monogamous behavior as given by inequality~(\ref{mon2}), with saturation 
achieved for this state. 

\section{Illustrations}

\subsection{Pure states with Schmidt decomposition}

Let us illustrate the monogamy bound~(\ref{monogamybound}) 
in the case of multipartite pure states $|\psi\rangle$ that admit Schmidt decomposition, whose explicit conditions of 
existence are discussed in Ref.~\cite{SchDec}. We will assume that the system is composed by a set of qubits. 
In such a case, we can write 
$|\psi\rangle = \sum_{i=1}^{2} \sqrt{p_i} |i_{A_1}\rangle \otimes \cdots \otimes |i_{A_N}\rangle$, where 
$\{|i_{A_k}\rangle\}$ are orthonormal bases, $p_i \ge 0$, and $\sum_i p_i = 1$.
For the density operator $\hat{\rho}_{A_1 \cdots A_N}=|\psi\rangle \langle \psi |$, we obtain
\begin{equation}
\hat{\rho}_{A_1 \cdots A_N} = \sum_{i,j=1}^2 \sqrt{p_i p_j} |i_{A_1} \cdots i_{A_N}\rangle \langle j_{A_1} \cdots j_{A_N}| .
\end{equation}
Since Schmidt decomposition implies equal spectrum for all single-qubit reduced density operators $\hat{\rho}_{A_k}$, we obtain 
that $S(\hat{\rho}_{A_k}) = -\sum_{i=1}^2 p_k \log_2 p_k \equiv S $, for any individual subsystem $A_k$. Therefore, 
the mutual information is $I(\hat{\rho}_{A_1 \cdots A_N}) = N \, S$. In order to consider measurements 
$\Phi\left( \hat{\rho}_{A_1 \cdots A_N} \right)$ over $\hat{\rho}_{A_1 \cdots A_N}$, it can be shown that, 
by adopting projective (von Neumann) measurements, the minimization of the loss of correlation is obtained in Schmidt 
basis, namely, $\{\hat{\Pi}_{A_k}^{i}\} = \{|i_{A_k}\rangle \langle i_{A_k}|\}$. This is a consequence of both the 
group homomorphism of $U(2)$ to $SO(3)$ and the monotonicity of entropy under majorization (see discussion for 
the state $(|0 \cdots 0 \rangle + |1 \cdots 1 \rangle)/\sqrt{2}$ in Ref.~\cite{Xu:12}). 
Then, $\Phi\left( \hat{\rho}_{A_k} \right) = \hat{\rho}_{A_k}$, which implies 
$S(\Phi\left( \hat{\rho}_{A_k} \right)) = S$. Moreover
$\Phi(\hat{\rho}_{A_1 \cdots A_N}) = \sum_{i=1}^2 p_i |i_{A_1} \cdots i_{A_N}\rangle \langle i_{A_1} \cdots i_{A_N}|$.
Therefore, the mutual information after measurement is $I(\Phi(\hat{\rho}_{A_1 \cdots A_N})) = (N-1) \, S$, which  
implies that 
$\mathcal{D}\left( A_1 : \cdots :A_N \right) = S$.
Therefore, the sum of pairwise QDs is upper limited by the von Neumann entropy $S$ of an individual subsystem, which holds for quantum systems 
composed of an arbitrary number $N$ of qubits. As an example, consider the GHZ state for $N$ qubits, namely, 
$|GHZ_N\rangle = \left(|0 \cdots 0 \rangle - |1 \cdots 1 \rangle \right)/\sqrt{2}$.
For this state, we obtain $\mathcal{D}\left( A_1 : \cdots :A_N \right) = 1$ (independently of $N$) and vanishing pairwise 
correlations $\mathcal{D}\left( A_1 : A_k \right)$, which is in agreement with Theorem~\ref{t3}. 

\begin{figure}[th]
\centering {\includegraphics[angle=0,scale=0.3]{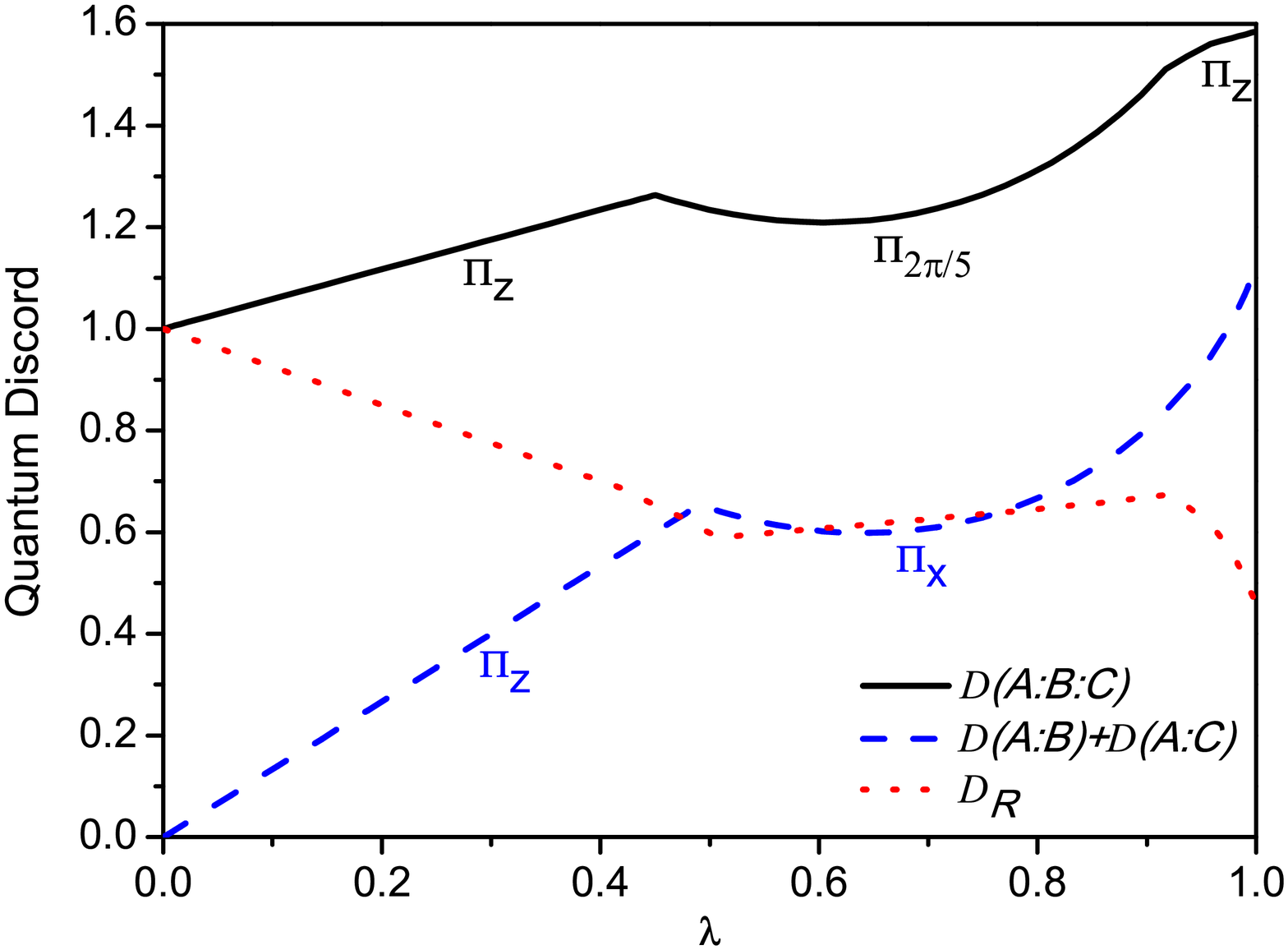}}
\caption{(Color online) GQD, pairwise QD, and residual QD for the W-GHZ mixed state as a function of the parameter $\lambda$. For each 
curve, the minimizing bases $\{\Pi_{A_k}\}$ are indicated, which are equal for all qubits measured. As $\lambda$ varies, the minimizing 
bases are denoted in terms of the Bloch sphere angles ($\theta, \phi$), with $\Pi_z$, $\Pi_x$, and $\Pi_{2\pi/5}$ associated with the angles 
$(0,0)$, $(\pi/2,0)$, and $(2\pi/5,0)$, respectively. }
\label{f2}
\end{figure}

\subsection{Tripartite W-GHZ mixed state}

Let us now consider a tripartite system $ABC$ described by the W-GHZ state
\begin{equation}
\hat{\rho}=\lambda \left\vert W\right\rangle \left\langle W\right\vert
+\left( 1-\lambda \right) \left\vert GHZ\right\rangle \left\langle
GHZ\right\vert ,
\end{equation}
where
$\left\vert W  \right\rangle =\left( \left\vert 001 \right\rangle + 
\left\vert 010 \right\rangle +\left\vert 100 \right\rangle \right)/\sqrt{3}$ and
$\left\vert GHZ \right\rangle =\left( 
|000\rangle - |111\rangle \right)/\sqrt{2}$, with $0 \le \lambda \le 1$.
Note that, for $\lambda =0$ and $\lambda =1$, we have pure states, given by $|GHZ\rangle$ and 
$|W\rangle$ states, respectively. Therefore, by adopting projective measurements, we will have 
that $\mathcal{D}\left( A: B : C \right) = 1$
for $\lambda=0$. As $\lambda$ increases, we numerically find out a monotonic increase of GQD until 
$\mathcal{D}\left( A : B: C \right) = \log_2 3 $ for $\lambda=1$. This can be seen as a consequence of 
the absence of Schmidt decomposition for the $W$ state, which leaves GQD unconstrained by the entropy of an individual subsystem. 
For the complete range of $\lambda$, we plot GQD in Fig.~\ref{f2} as well as the pairwise sum $\mathcal{D}\left( A : B \right)$ + 
$\mathcal{D}\left( A : C \right)$. We observe that monogamy as given by inequality~(\ref{mon2}) 
is obeyed for any $\lambda$, with distinct minimizing bases $\{\Pi_{A_k}\}$ for GQD and pairwise QD. Note also from the plot 
that the bound is considerably tight in the sense that, for any $\lambda$, the value of GQD is sufficiently limited to ensure that the 
subsystems are {\it not} freely correlated. 

\section{Residual multipartite correlations}

The monogamy bound~(\ref{mon2}) allows for the definition of a measure 
$\mathcal{D}_{R}$ for {\it residual} multipartite quantum correlations, 
namely, contributions to quantum correlations beyond pairwise QD. 
This is similar to the definition of {\it tangle} as a measure of residual multipartite entanglement~\cite{Coffman:00}. 
Indeed, let us consider the particular case of tripartite permutation invariant states, such the W-GHZ state. 
In such a case, we can define the residual QD as   
\begin{equation}
\mathcal{D}_{R} \equiv 
\mathcal{D}\left( A : B : C \right) - \mathcal{D}\left( A : B \right) - \mathcal{D}\left( A : C \right).
\end{equation}
Note that, by considering the non-increasing behavior of QD under discard of subsystems (such as in the W-GHZ state), 
monogamy is ensured, implying that $\mathcal{D}_{R} \ge 0$. If permutation invariance is absent, 
the tripartite residual correlations can be defined by $\min \mathcal{D}_{R}$, where minimization is taken over all subsystem permutations 
(see Ref.~\cite{Adesso:07} for a similar definition in the case of entanglement). 
Applying the residual measure for the GHZ state, it follows that $\mathcal{D}_{R}^{GHZ} = 1$, since pairwise contributions vanish. For the 
W state, residual QD is lower than in the case of the GHZ state, since pairwise QD is nonvanishing. Indeed, it can be 
be shown that $\mathcal{D}^{W}\left( A : B \right) = \mathcal{D}^{W}\left( A : C \right) = 2\log _{2}\left( 3\right) -\frac{2}{3%
}-\frac{5}{6}\log _{2}\left( 5\right) \approx 0.568$ (with the minimizing measurement 
found in the $\sigma_x$ eigenbasis). This implies that $\mathcal{D}_{R}^{W} =  \frac{4}{3}+\frac{5}{3}\log _{2}\left( 5\right) -3\log _{2}\left( 3\right)  \approx 0.448$. 
This behavior is exhibited in Fig.~\ref{f2} and is in agreement with the hierarchy found in Ref.~\cite{Giorgi-2:11}. 

\section{Conclusions}

In conclusion, we have introduced a monogamy bound for pairwise quantum correlations based on the amount of GQD available to 
a multipartite system. Remarkably, this monogamy inequality holds for general states whose bipartite quantum discord is non-increasing 
under discard of subsystems.  In particular, oppositely to the typical monogamy bound~(\ref{mon1}), it covers both GHZ and W tripartite 
classes of states, providing therefore a promising setup for the investigation of measures for multipartite classical and quantum residual 
correlations. We leave such topics for further research.  

\begin{acknowledgments}

We thank Kavan Modi and J{\'o}zsef Pitrik for helpful discussions. 
This work is supported by CNPq, CAPES, FAPERJ, and the Brazilian National 
Institute for Science and Technology of Quantum Information (INCT-IQ).

\end{acknowledgments}


\end{document}